\documentclass[letterpaper, 10 pt, conference]{ieeeconf}
\IEEEoverridecommandlockouts 
\usepackage{graphicx} 
\usepackage{cite}
\usepackage{caption}
\usepackage{amsmath}
\usepackage{amsfonts}
\usepackage{mathtools} 
\usepackage{dsfont}
\usepackage{color}
\usepackage{float}
\let\labelindent\relax
\usepackage{enumitem,kantlipsum}
\usepackage{subcaption}
\usepackage{bm}
\usepackage{amsmath,amssymb}

\newtheorem{proposition}{Proposition}

\newtheorem{remark}{Remark}

\date{} 

\title{\LARGE \bf Robust Dynamic Pricing and Admission Control \\
with Fairness Guarantees}

\author{Yingqing Chen, Anni Li, Christos G. Cassandras, Homayoun Hamedmoghadam, Fabian Wirth, and Robert Shorten
\thanks{Y. Chen and C. G. Cassandras are
with the Division of Systems Engineering and Center for Information and
Systems Engineering, Boston University, Brookline, MA 02446
\tt\small\{yqchenn;cgc\}@bu.edu.}
\thanks{Anni Li is with the Department of ECE, UNC-Charlotte, NC \texttt{{\small ali20@charlotte.edu}} }
\thanks{R. Shorten is with the Dyson School of Design Engineering, Imperial College London, London, UK.  \texttt{{\small r.shorten@imperial.ac.uk}} }
\thanks{H. Hamedmoghadam is with the Department of Electrical \& Electronic Engineering, Imperial College London, London, UK.  \texttt{{\small h.hamed@imperial.ac.uk}} }
\thanks{F. Wirth is with the Department of Mathematics, University of Passau, Passau, Germany. \texttt{{\small fabian.wirth@uni-passau.de}}}
}

\begin{document}
\maketitle
\thispagestyle{empty}
\pagestyle{empty}

\begin{abstract} 
 Dynamic pricing is commonly used to regulate congestion in shared service systems. This paper is motivated by the fact that in the presence of users with varying price sensitivity (responsiveness), conventional monotonic pricing can lead to unfair outcomes by disproportionately excluding price-elastic users, particularly under high or uncertain demand. We therefore develop a fairness-oriented mechanism under demand uncertainty.
The paper's contributions are twofold. First, we show that when fairness is imposed as a hard state constraint, the optimal (revenue maximizing) pricing policy is generally non-monotonic in demand. This structural result departs fundamentally from standard surge pricing rules and reveals that price reduction under heavy load may be necessary to maintain equitable access.
Second, we address the problem that price elasticity among heterogeneous users is unobservable. To solve it, we develop a robust dynamic pricing and admission control framework that enforces capacity and fairness constraints for all user type distributions consistent with aggregate measurements. 
By integrating integral High Order Control Barrier Functions (iHOCBFs) with a robust optimization framework under uncertain user-type distribution, we obtain a controller that guarantees forward invariance of safety and fairness constraints while optimizing revenue.
Numerical experiments demonstrate improved fairness and revenue performance relative to monotonic surge pricing policies.

\end{abstract}




\section{Introduction}
Dynamic pricing, commonly implemented in practice through surge pricing, is used to manage access to congestible service systems such as ride-hailing platforms, shared mobility services, healthcare scheduling, and online marketplaces \cite{Onita2023SurgePricing}. By adjusting prices in response to demand fluctuations, service providers can influence user participation, regulate congestion, maximize revenue, and improve overall system efficiency \cite{bimpikis2019spatial,besbes2021surgepricing,hu2022surge}. A substantial literature has shown that appropriate pricing policies can stabilize queueing systems and improve performance, including static pricing rules with robustness properties \cite{bergquist2025static} and dynamic pricing and matching for two-sided or shared-resource systems \cite{varma2023dynamic,banerjee2022pricing}. Consistent with these benefits, evidence from airline revenue management suggests that dynamic pricing can increase revenue by 1 to 4 percent relative to traditional pricing \cite{wittman2018dynamic}.

Despite these advantages, dynamic pricing has raised significant practical concerns: aggressive price discrimination may intensify competition and reduce revenue, fares may spike during emergencies, and regulators have warned that algorithmic pricing, while potentially improving efficiency, can also amplify competitive concerns \cite{wittman2018dynamic,Onita2023SurgePricing,CMA2021Algorithms}.
It also raises important fairness concerns: congestion-based price increases may disproportionately disadvantage price-sensitive users, especially during high demand periods. In socially essential services such as transportation and healthcare, such exclusion directly undermines equitable access and social welfare \cite{seele2021ethicality}.
Even when uniform pricing is imposed across protected groups, fairness is not guaranteed: heterogeneous price sensitivities, congestion tolerances, and dropout behavior may still generate systematic disparities in realized service access\cite{cohen2022price,king2025dynamic}.

These concerns have motivated growing interest in fairness-aware dynamic pricing. Formal notions of fairness, including group and individual fairness, have been adapted from algorithmic decision-making settings \cite{dwork2012fairness} and incorporated into pricing models through constraints on price dispersion or access rates \cite{cohen2022price,xu2023doubly,cohen2025dynamic}. 
Other formulations impose fairness constraints within online learning or regret minimization frameworks \cite{xu2023doubly,cohen2025dynamic}, while related queueing-based pricing models study fair and efficient resource sharing under strategic behavior\cite{nakamura2025fair}.

However, two fundamental limitations remain.
First, fairness is typically enforced as an objective or in expectation, rather than as a real-time state constraint \cite{cohen2022price,cohen2025dynamic,xu2023doubly}. As a result, transient fairness violations can occur during demand surges. 
Second, most existing models assume that user class attributes are observable and can be directly incorporated into pricing decisions \cite{xu2023doubly,cohen2022price,cohen2025dynamic,nakamura2025fair}.
In practice, class-level states may only be indirectly inferred from aggregate system information, may be estimated with substantial noise, or may be unavailable because monitoring protected attributes, such as race, gender, or disability status, is legally or ethically constrained \cite{mozannar2020fair}. 
These gaps motivated our prior work \cite{king2025dynamic}, which studied dynamic pricing under unobservable demand composition and characterized fairness conditions in heterogeneous queueing systems, building on access-control models from communication networks \cite{budzisz2011fair}. In that setting, elastic and inelastic traffic compete for a congestible shared service under a common price signal. The analysis shows that monotonic surge pricing can drive out elastic users without alleviating congestion, whereas a non-monotonic pricing policy can sustain a desirable equilibrium with increased fairness. More broadly, the prior work provides a mathematical equilibrium analysis and further characterizes the domain of attraction of the desired equilibrium.

This paper builds on that foundation by moving from qualitative equilibrium analysis to online controller synthesis. Specifically, we jointly control the (a) price, which affects users’ decisions to enter, and (b) the admission rate, which regulates how much waiting demand is allowed into the service system. Both controls are computed in real time under explicit resource-capacity and fairness constraints while accounting for unobservable demand composition.

Our contributions are twofold. First, we identify a fundamental price-demand relationship in fairness-constrained pricing systems: when fairness is enforced as a hard state constraint, the revenue-maximizing price is generally non-monotonic in demand. 
Second, we develop a robust dynamic pricing and admission control framework for heterogeneous demand with unobservable composition. The resulting controller provides real-time fairness and safety guarantees while optimizing revenue. In this sense, it offers an effective alternative to conventional surge pricing, with simulations showing improved fairness and revenue under stochastic and time-varying demand.

\section{Problem Formulation}
As shown in Fig. \ref{fig:demo}, we consider a queueing system with the service demand comprised of $N$ classes of users sharing similar behavioral characteristics, such as price sensitivity or demand level sensitivity. 
Users awaiting admission are modeled by the $N$ \emph{demand queues} corresponding to the sensitivity levels.
The waiting users may exit the system before admission due to their sensitivity to price or to current demand queue length. 
The remaining users are gradually admitted into the service queue at an admission rate $\alpha(t)$. 

Once admitted to the \emph{service queue}, users are charged a one-time service price $p(t)$, which is uniform across all users. Users in the service queue are then served at a service rate $\mu(t)$, which is characteristic of the system and can be time-varying.

Let \(x_i(t)\), \(i = 1,\ldots,N\), denote the queue length of the \(i\)-th demand queue. Note, however, that while these demand queues conceptually exist, their individual lengths are not directly observable. Instead, only the aggregate demand queue length
\begin{equation} \label{eqn: z}
    z(t) = \sum_{i=1}^N x_i(t)
\end{equation}
is assumed to be observable, as it can be readily measured or estimated in practice. 

\begin{figure}
    \centering
    \includegraphics[width=0.99\columnwidth]{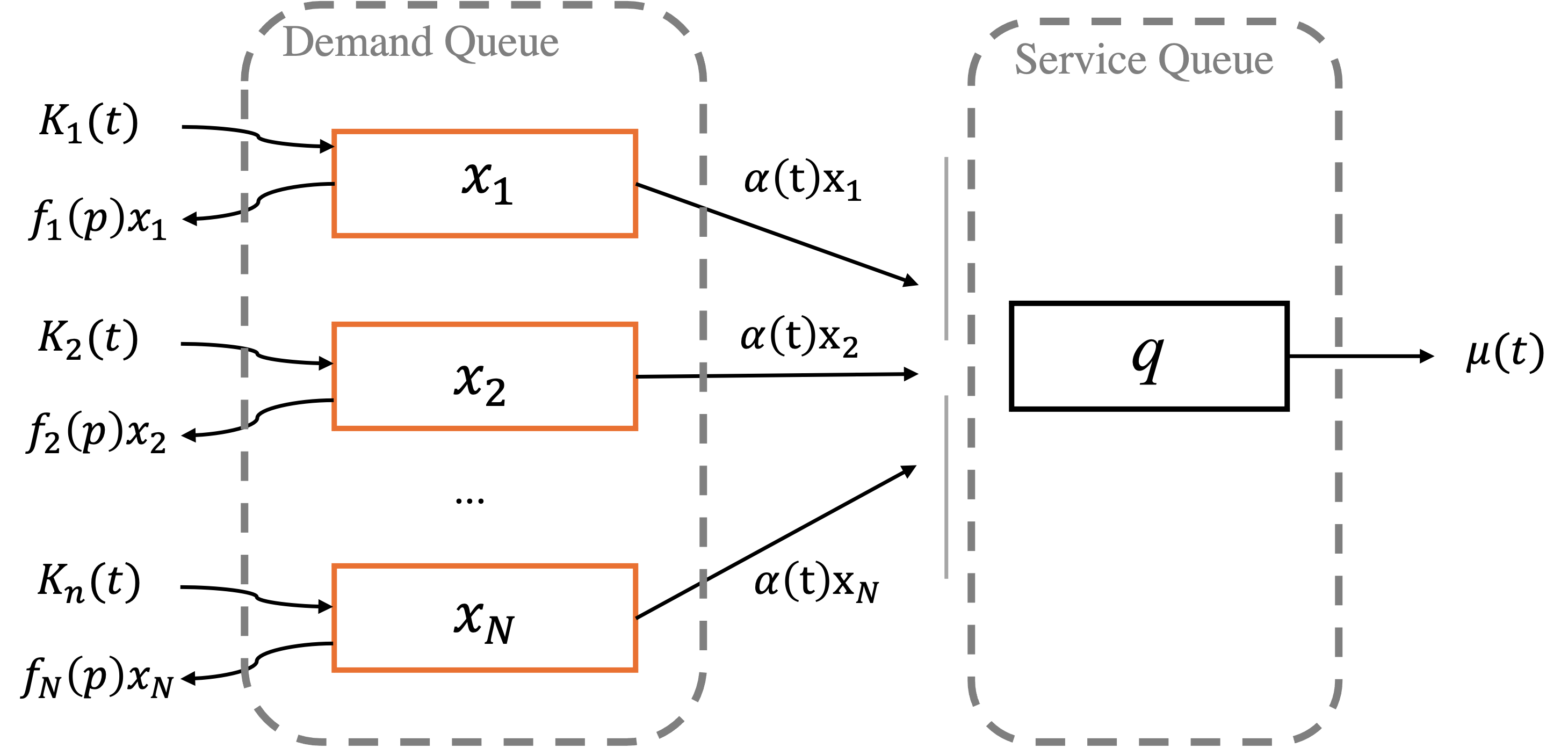}
    \caption{\small\textbf{Two-stage queueing model with heterogeneous demand sources.} 
   Each class $i$ ($i = 1, \ldots, N$) maintains a demand buffer $x_i$, 
    driven by exogenous arrivals $K_i(t)$ and price-dependent dropout $f_i(p)x_i$.
    A fraction $\alpha(t)$ of demand is admitted to the shared service queue $q$,
    which serves at rate $\mu(t)$.
     }

    \label{fig:demo}
\end{figure}

As shown in Fig.~\ref{fig:demo}, the dynamics of each demand queue \(x_i(t)\) are governed by three factors:
\begin{enumerate}
    \item \emph{Exogenous demand}: Each source generates a time-varying arrival process \(K_i(t)\), which is independent of the queue state and control actions.
    \item \emph{Dropout}: A fraction of users may abandon the demand queue due to instantaneous pricing or demand level.
    The dropout rate depends on the service price \(p(t)\) and is heterogeneous across demand sources. We model this effect using a source-specific dropout function \(f_i(p(t))\).
    \item \emph{Admission}: Users are admitted from the demand queues to the service queue at a rate \(\alpha(t)\), which may depend on the current service queue length \(q(t)\).

\end{enumerate}

The evolution of the service queue length \(q(t)\) is determined by:
\begin{enumerate}
    \item \emph{Admission}: The aggregate inflow from all demand queues.
    \item \emph{State-dependent service}: Users depart the service queue at a service rate \(\mu(t)\), which may depend on the current queue state \(q(t)\).
\end{enumerate}

The system controller
seeks to jointly determine the pricing policy \(p(t)\) and admission control \(\alpha(t)\) 
over time so as to:
$(a)$ Regulate congestion by ensuring that \(q(t)\) remains below a prescribed maximum level \(q^{\max}\),
$(b)$ Maximize the total revenue generated by the system, and 
$(c)$ Maintain fairness across heterogeneous demand classes.


Note that although incoming demand is conceptually divided into heterogeneous classes, the controller cannot differentiate among them when implementing control actions. Consequently, both the service price \(p(t)\) and the admission rate \(\alpha(t)\) are uniform across all demand sources.

\begin{remark}
    Queueing systems with a preliminary demand buffer and endogenous dropout are common in real-world applications. A representative example is a ride-hailing platform, where users wait in a virtual buffer before being matched with a driver and may abandon the system due to high prices or unfavorable demand queue conditions.
\end{remark}

\begin{remark}
    The importance of fairness depends on both the nature of the service system and its operating regime. In essential service, 
    such as healthcare programs or public transportation networks, equitable access is often a primary requirement and should therefore be enforced explicitly and in real time. In private platforms, such as ride-hailing markets, revenue maximization may be the dominant objective under normal conditions. Still, the revenue-fairness trade-off can change substantially across different demand and congestion regimes: under some conditions the two objectives may be closely aligned, while under others they may conflict more sharply. Our framework is designed to accommodate this variation by treating fairness as a tunable operational requirement rather than assuming the same priority structure across all systems and states.
    As shown later in the simulation study, the unfairness threshold provides a direct mechanism for tuning this trade-off between fairness and revenue.
\end{remark}


\textbf{System Dynamics.}
In order to formulate an optimization problem, we model the dynamics with generality while maintaining analytically tractable. Each of the $N$ user classes is characterized by a distinct exogenous demand process. The control inputs $p(t)$ and $\alpha(t)$ represent pricing and admission mechanisms applied uniformly across all user classes.

User heterogeneity is captured through class-dependent dropout behavior. 
In practice, users may differ in their ability to comply with higher prices, so the same offered price can push some classes of users out of the system more quickly than others. 
Therefore, differences in dropout become an interpretable marker of the inequities that can arise even when one price is applied to everyone.
For analytical simplicity, the dropout rate of each class is modeled as a linear function of the price, though the framework readily extends to more general non-decreasing functions. For each class \(i = 1,2\ldots N\), the dropout rate is given by
\begin{equation}\label{eqn:dropout_rate}
    f_i(p(t)) = r_{i,1} p(t) + r_{i,2},
\end{equation}
where \(r_{i,1}\) characterizes the price sensitivity of class \(i\), and \(r_{i,2}\) captures its baseline dropout tendency independent of price. The parameters are normalized to ensure \(0 \leq f_i(p(t)) \leq 1\). 

The dynamics of the demand queue associated with each class \(i\) are described by
\begin{equation}\label{eqn:demand_queue_dynamics}
    \dot{x}_i(t) = K_i(t) - f_i(p(t)) x_i(t) - \alpha(t) x_i(t),
\end{equation}
where \(K_i(t)\) denotes the exogenous demand arrival rate of class \(i\), and \(\alpha(t)\) is the admission rate to the service queue.

Due to partial observability, we use total demand arrival rate $K(t) = \sum_i K_i(t)$,  and time-varying proportion variables $w_i(t)=z(t)/x_i(t)$. 
The aggregate demand queue dynamics can be written as
\begin{equation}\label{eqn: dynamics_total_demand_queue}
    \dot{z}(t) = K(t)
    - \sum_{i=1}^N f_i(p(t)) w_i(t) z(t)
    - \alpha(t) z(t).
\end{equation}

The dynamics of the service queue are then given by
\begin{equation}\label{eqn: dynamics_service_queue}
    \dot{q}(t) = \alpha(t) z(t) - \mu(t),
\end{equation}

We use the same service rate function as in \cite{king2025dynamic}, assuming that the service rate increases linearly with the queue length when the system is lightly loaded and saturates at a maximum capacity. This choice mimics the nature of many real-world systems with bounded service capacity.
Specifically, given a constant maximum service-rate \(\mu^* > 0\) and critical queue-length \(q_c > 0\) intrinsic to a system, the service rate can be defined as
\begin{equation}\label{eqn:service_rate}
\mu(t) =
\begin{cases}
\dfrac{\mu^*}{q_c} \, q(t), & 0 \leq q(t) \leq q_c, \\[6pt]
\mu^*, & q(t) \geq q_c.
\end{cases}
\end{equation}

\subsection{Fairness Level}
One of the primary control objectives is to maintain fairness across heterogeneous demand groups. 
To quantify fairness, we introduce an \emph{unfairness index}:
\begin{equation}\label{eqn: Unfairness Index}
I(t) \;=\; \frac{1}{T_I} \int_{t-T_I}^{t}
\frac{\sum_{i} f_i\!\left(p(\tau)\right)\, w_i(\tau)\, z(\tau)}
{K(\tau)} \, d\tau .
\end{equation}

The index \(I(t)\) is the time-averaged ratio of the total dropout volume (which corresponds to users whose requests are not fulfilled) and the total observed demand over the window $T_I$. A larger value of \(I(t)\) indicates a higher degree of unfairness in service allocation, while a smaller value reflects a more balanced and equitable system performance.
The window length \(T_I\) serves as a tunable parameter that determines the temporal sensitivity of the fairness metric. A larger \(T_I\) smooths short-term fluctuations and emphasizes long-term fairness, thereby allowing transient unfairness, while a smaller \(T_I\) enforces stricter regulation by penalizing even short-term deviations from fairness.

\begin{remark}
    The fairness metric in \eqref{eqn: Unfairness Index} is not intended to be the only possible choice. Rather, it serves as a practical representative that allows us to develop a control framework with explicit fairness guarantees. The same framework can accommodate alternative or enriched definitions of fairness, including metrics that account for additional service-quality factors such as waiting time prior to dropout.
\end{remark}

\subsection{Optimal Control Problem}
Our objective is to determine the optimal price and admission rate that \emph{jointly} maximize the total revenue, subject to two constraints:  
(a) the service queue length must remain below its capacity ($q(t) \leq q_{\max}$); and  
(b) the unfairness level must be maintained within an acceptable predefined range.

Although a dynamic adjustment of the price and admission rate is desirable, excessively frequent policy updates undermine users’ trust in the pricing policy. To address this issue, we introduce a decision time window of length \(T_d\). At each decision point, an optimal control policy is computed based on the most recent system observations and is held constant over the subsequent time window \([t, t+T_d]\). 
Within each decision window, we maximize the total revenue generated over that interval. This objective reflects the natural operational goal of the service provider, while service capacity and fairness are enforced separately as hard constraints. 
The parameter \(T_d\) determines the frequency of policy updates and reflects a trade-off between \emph{responsiveness} and \emph{stability}. This receding-horizon implementation preserves the adaptivity of the dynamic pricing mechanism while mitigating instability induced by rapid policy fluctuations.


\textbf{Problem 1: Revenue Maximization with Capacity and Fairness Constraints.}
At each decision point, determine the pricing and admission policy over the next decision window $[t, t+T_d]$ so as to maximize the total revenue, subject to service capacity and fairness requirements.

\emph{Objective:}  
Maximize the total revenue accumulated over the upcoming decision window. The instantaneous revenue is given by the product of the price $p(t)$, admission rate $\alpha(t)$, and total observed demand $\zeta(t)$. The resulting objective function is
\begin{equation}\label{eqn: obj}
    \max \; J(t) = \int_{t}^{t+T_d} p(t)\, \alpha(t)\, z(\tau)\, d\tau . 
\end{equation}


\emph{Constraint 1 (Service Queue Capacity):}  
The service queue length $q(t)$ must not exceed its allowable capacity $q_{\max}$.
This requirement is expressed as
\begin{equation}\label{eqn: const1_service queue capacity}
    b(q(t)) = q_{\max} - q(t) \geq 0 .
\end{equation}

\emph{Constraint 2 (Fairness):}  
To prevent excessive unfairness among users, the unfairness index $I$ is required to remain below a prescribed threshold $\theta_d$ at all times: 
\begin{equation}\label{eqn: const2_fairness_constraint}
    I(t) \leq \theta_d ,
\end{equation}
where $\theta_d$ denotes the maximum allowable level of unfairness.

This completes the formulation of the nominal optimal control problem. In the next section, we introduce reformulations to address state-constraint enforcement and partial observability issues.

\section{Robust Dynamic Pricing Control with iHOCBF}
Problem~1 cannot be solved directly for two reasons. First, the service queue capacity constraint cannot be enforced explicitly through the control inputs, as the queue length evolves according to system dynamics rather than being directly actuated.
In other words, 
(\ref{eqn: const1_service queue capacity}) constrains the system state, but does not provide information about the control that can enforce the coinstraint.
Second, the aggregate demand queue dynamics in \eqref{eqn: dynamics_total_demand_queue} depend on the proportion variables \(w_i(t)\), which are not directly observable. This partial observability induces parametric uncertainty in the system dynamics and prevents exact state feedback implementation. 
In this section, we address both challenges and develop a tractable solution approach
by enforcing queue capacity via High Order Control Barrier Functions (HOCBFs) and by introducing a robust control framework that accommodates uncertainty in the demand composition.

\subsection{CBF-Based Enforcement of Service Queue Capacity Constraints}

Among the two constraints in Problem~1, \emph{Constraint~2} depends explicitly on the control inputs and can therefore be imposed directly. In contrast, \emph{Constraint~1} defines a safe region in the state space, so it must be enforced indirectly by determining a control that properly steers the system dynamics. To do this, we transform the state constraint into a control-dependent safety condition using Control Barrier Functions (CBFs) \cite{xiao2023safe}.

Briefly, for a control-affine system \(\dot{x}=f(x)+g(x)u\), let \(b(x)\geq 0\) characterize the safe set. A CBF replaces the original state constraint with an inequality involving the control input \(u\). If this inequality is satisfied at all times and the initial state lies in the safe set, then the safe set remains \emph{forward invariant}.
Note that this is a \emph{sufficient} condition, hence this inequality possibly introduces some conservativeness.
A key concept is the \emph{relative degree} of \(b(x(t))\) with respect to the control input, defined as the number of times \(b(x(t))\) must be differentiated along the system dynamics before \(u\) appears explicitly \cite{xiao2023safe}. When the relative degree is one, a standard CBF is sufficient. When it exceeds one, however, a high-order CBF (HOCBF) is required \cite{xiao2021high}. Specifically, suppose \(b(x(t))\) has relative degree \(r\), where \(r\) is the smallest integer such that \(L_gL_f^{r-1}b(x)\neq 0\). Define
\begin{equation}
\begin{aligned}
\psi_0(x)&=b(x), \\
\psi_i(x)&=\dot{\psi}_{i-1}(x)+\alpha_i(\psi_{i-1}(x)), \quad i=1,\dots,r-1,
\end{aligned}
\end{equation}
where each \(\alpha_i\) is a class-\(\mathcal{K}\) function, and \(L_f(\cdot)\) and \(L_g(\cdot)\) denote Lie derivatives along \(f\) and \(g\), respectively. The corresponding HOCBF constraint is \cite{xiao2023safe}
\begin{equation}\label{eqn: general_hocbf_constraint}
L_f^r b(x)+L_gL_f^{r-1}b(x)\,u+O\bigl(b(x)\bigr)+\alpha_r\!\left(\psi_{r-1}(x)\right)\geq 0,
\end{equation}
where \(O(b(x))\) collects the lower-order Lie derivative terms generated when differentiating \(\psi_{r-1}\). This formulation converts the original safety requirement into an affine inequality in the control input while preserving forward invariance, and it provides a sufficient condition for guaranteeing satisfaction of the original safety constraint.

In our model, \eqref{eqn: dynamics_total_demand_queue} and \eqref{eqn: dynamics_service_queue} imply that the service queue capacity constraint has different relative degrees with respect to the two control variables. The admission rate \(\alpha(t)\) appears after one differentiation, whereas the pricing control \(p(t)\) appears only after two differentiations.
This mismatch creates a structural limitation. If the safety constraint \eqref{eqn: const1_service queue capacity} were enforced with a standard CBF, the resulting condition would depend only on \(\alpha(t)\). The controller could then satisfy the queue-capacity constraint by adjusting the admission rate alone, leaving pricing effectively disconnected from congestion regulation.

To overcome this issue, we adopt an integral HOCBF (iHOCBF) formulation as in \cite{xiao2022ihocbf}, which augments the system dynamics so that the service queue capacity constraint has the same relative degree with respect to both control inputs. As a result, both pricing and admission decisions remain active in enforcing the safety constraint, while the optimization problem remains tractable and the forward invariance of the safe set is preserved.

\subsubsection{Extended State Space}

To resolve the relative-degree mismatch identified earlier, we augment the system with auxiliary dynamics for the admission rate \(\alpha(t)\), treating it as a state variable: 
\begin{equation}\label{eqn: aux_alpha}
    \dot{\alpha}(t) = \nu(t),
\end{equation}
where \(\nu(t)\) is a new control input.
The augmented system is thus controlled by the input pair \((p(t), \nu(t))\), which jointly influence the demand and service dynamics through pricing and admission-rate regulation, respectively.
Thus, we define the extended state vector
\begin{equation}\label{eqn: extended state vector}
    \mathbf{x}(t) = \begin{bmatrix} q(t) & z(t) & \alpha(t) \end{bmatrix}^\top,
\end{equation}
which captures the service queue length, the aggregate demand queue length, and the admission rate, respectively.
The corresponding system dynamics can be written in affine control form as follows (for notational simplicity, we omit the explicit time index \(t\) of some variables):
\begin{equation}\label{eqn: dynamics of new states}
\small
\dot{\mathbf{x}} =
\underbrace{\begin{bmatrix}
\alpha z - \mu \\
K - z \mathbf{w}^\top \mathbf{r}_2 - \alpha z \\
0
\end{bmatrix}}_{F(\mathbf{x})}
+
\underbrace{\begin{bmatrix}
0 \\
- z \mathbf{w}^\top \mathbf{r}_1 \\
0
\end{bmatrix}}_{G_p(\mathbf{x})} p
+
\underbrace{\begin{bmatrix}
0 \\
0 \\
1
\end{bmatrix}}_{G_\nu(\mathbf{x})} \nu .
\end{equation}
Here, 
the vector \(\mathbf{w} = [w_1, \ldots, w_N]^\top\) represents the proportions of each user class in the demand queue,  satisfying \(\sum_{i=1}^N w_i = 1\).
The parameter vectors \(\mathbf{r}_1 = [r_{1,1}, \ldots, r_{N,1}]^\top\) and \(\mathbf{r}_2 = [r_{1,2}, \ldots, r_{N,2}]^\top\) collect the class-specific price sensitivity and baseline dropout coefficients, respectively, as defined in \eqref{eqn:dropout_rate}.

\subsubsection{iHOCBF Formulation}
Let the service queue capacity constraint be expressed as a  function of the extended state vector
\begin{equation} \label{eqn: new b(x)}
    b(\mathbf{x}) = q_{\max} - q .
\end{equation}
Under the extended dynamics \eqref{eqn: dynamics of new states}, the constraint \(b(\mathbf{x}) \ge 0\) has maximum relative degree two. 
We therefore employ an iHOCBF constraint with the form of \eqref{eqn: general_hocbf_constraint} to enforce forward invariance:
\begin{equation}\label{eqn: iHOCBF}
\begin{split}
\eta_1(\mathbf{x}, p, \nu,\mathbf{w}) &= \mathcal{L}_F^2 b(\mathbf{x}) + \mathcal{L}_{G_p} \mathcal{L}_F b (\mathbf{x}) p + \mathcal{L}_{G_\nu} \mathcal{L}_F b (\mathbf{x}) \nu \\
&\quad + \lambda_2 [ \mathcal{L}_F b(\mathbf{x}) +\mathcal{L}_{G_p}b(\mathbf{x}) p + \mathcal{L}_{G_{\nu}}b(\mathbf{x}) \nu  \\
&\qquad + \lambda_1 (b(\mathbf{x}))]\geq 0
\end{split}
\end{equation}
where linear class-\(\mathcal{K}\) functions are used for simplicity, and \(\lambda_1, \lambda_2 > 0\) are the corresponding coefficients for first and second order respectively.

The required Lie derivatives are computed as follows:
\begin{equation}\label{eqn:LieF}
\mathcal{L}_F b(\mathbf{x}) = - \alpha z + \mu,
\end{equation}
\begin{equation}\label{eqn:LieGp}
\mathcal{L}_{G_p} b(\mathbf{x}) = 0,
\qquad
\mathcal{L}_{G_\nu} b(\mathbf{x}) = 0,
\end{equation}
\begin{equation}\label{eqn:LieF2}
\mathcal{L}_F^2 b(\mathbf{x})
= (\alpha z - \mu) \frac{d\mu}{dq}
- \alpha \big[ K - z (\mathbf{w}^\top \mathbf{r}_2) - \alpha z \big],
\end{equation}
\begin{equation}\label{eqn:LieGpF}
\mathcal{L}_{G_p} \mathcal{L}_F b(\mathbf{x})
= \alpha z (\mathbf{w}^\top \mathbf{r}_1),
\end{equation}
\begin{equation}\label{eqn:LieGnuF}
\mathcal{L}_{G_\nu} \mathcal{L}_F b(\mathbf{x})
= - z .
\end{equation}

This iHOCBF formulation explicitly incorporates both pricing and admission controls into the constraint, thereby ensuring forward invariance of the service queue capacity while preserving full actuation authority and computational tractability.

To better align the problem formulation with the discrete decision-making structure, we discretize time into uniform intervals of length \(T_d > 0\). Let \(k \in \mathbb{Z}_{\ge 0}\) denote the discrete decision index, with decision times \(t_k = k T_d\). The system state and control inputs evaluated at these decision points are denoted by
$\mathbf{x}_k = \mathbf{x}(t_k),~
p_k = p(t_k),~
\nu_k = \nu(t_k),~
\mathbf{w}_k = \mathbf{w}(t_k),~
I_k = I(t_k)$.
At each decision point \(t_k\), the $p_k$ and $\alpha_k$ are computed and held constant over the interval \([t_k, t_{k+1})\).

We can now replace the original capacity constraint with the iHOCBF condition and state the following discrete-time optimal control problem at each decision epoch \(t_k = kT_d\).

\textbf{Problem 2}:
Given the observed state \(\mathbf{x}_k = \mathbf{x}(t_k)\) and demand estimate over the next window, determine constant controls \((p_k,\nu_{k-1})\) that maximize the revenue over one decision window while enforcing queue-capacity safety and fairness:
\begin{subequations}\label{prob:2}
\begin{align}
\max_{p_k,\,\nu_{k-1}}\quad 
& J_k \;=\;   T_d\,p_k\,\alpha_{k}\,z_{k}\
\label{prob:2_obj}\\[1mm]
\text{s.t.}\quad 
& \eta_1\!\left(\mathbf{x}_{k+1},\,p_k,\,\nu_{k-1},\,\mathbf{w}_{k+1}\right)\;\ge\;0,
\label{prob:2_ihocbf}\\
& \eta_2(\mathbf{x}_{k+1},\,p_k,\,\nu_{k-1},\,\mathbf{w}_{k+1}) = \theta_d - I_{k+1} \geq 0
\label{prob:2_fairness}\\
& \eqref{eqn:service_rate}, \eqref{eqn: Unfairness Index}, \eqref{eqn: aux_alpha}, \eqref{eqn: dynamics of new states}, \eqref{eqn: iHOCBF} 
\label{prob:2_defs}
\end{align}
\end{subequations}

\subsection{Robust Pricing Control Framework}\label{sec: robust pricing control}

Since user types are not directly observable, uncertainty arises in the composition of the demand queue in Fig.~\ref{fig:demo}, even when the dynamics associated with each user type are accurately known.
In particular, the proportion vector $\mathbf{w}$ in \textbf{Problem 2} describing the distribution of user types in the demand queue is unknown.
In contrast, key aggregate quantities, including the total demand queue length ($z$), the total demand arrival rate ($K$), and the aggregate dropout quantity, which is denoted as $d$, are directly observable. We leverage these measurements to construct a constrained uncertainty set for \(\mathbf{w}\) and formulate a worst-case robust control framework. Specifically, we characterize the most adverse user-type distribution that is consistent with the observed data and enforce safety and fairness constraints under this worst-case realization. By guaranteeing constraint satisfaction in this setting, robust performance is ensured despite partial observability of the demand composition. 

For the \emph{Service Queue Capacity Constraint}, given the system state $\mathbf{x}$, price $p$, and auxiliary control input $\nu$, we define the associated worst-case safety margin by solving \textbf{Problem~3}:

\begin{align}\label{eqn: eta_1_star}
\eta_1^*(\mathbf{x}, p, \nu)
    \triangleq
    \min_{\mathbf{w}} \;& \eta_1(\mathbf{x}, p, \nu, \mathbf{w}) \\
\text{s.t.}\quad 
    \nonumber & \mathbf{1}^\top \mathbf{w} = 1, \\ \nonumber
    & \mathbf{w} \geq \mathbf{0}, \\ \nonumber
    & \sum_{i} f_i(p)\, w_i\, z = d
\end{align}
where $\mathbf{w}$ denotes the (unobservable) user-type proportion vector, $f_i(p)$ is the dropout rate of user type $i$ under price $p$, $z$ is the aggregate demand queue length, and $d$ is the observed aggregate dropout quantity. 
We call $\eta_1^*$ the \emph{worst-case} safety margin because it represents the minimal $\eta_1$ over all user-type distributions consistent with observed data. Ensuring $\eta_1^* \ge 0$ guarantees that \eqref{prob:2_ihocbf} is satisfied robustly for any feasible $\mathbf{w}$ under demand composition uncertainty.

Observe that when $N=2$, \textbf{Problem~3} need not be solved explicitly, since $\mathbf{w}$ can be determined directly from the relationship between $d$ and $z$. In contrast, when more than two user types are present, the vector $\mathbf{w}$ cannot in general be uniquely identified from the available measurements. 

Similarly, for the \emph{Fairness Constraint}, we define the associated worst-case fairness margin as
\begin{equation}\label{eqn: eta_2_star}
\eta_2^*(\mathbf{x}, p, \nu)
\triangleq
\min_{\mathbf{w}} \; \eta_2(\mathbf{x}, p, \nu, \mathbf{w}),
\end{equation}
subject to the same feasibility constraints on $\mathbf{w}$ as in \textbf{Problem~3}.

By including these new constraint, we can now define \textbf{Problem 4} for solving the robust optimal control problem when all constraints are guaranteed to be satisfied regardless of unobservable demand proportion:
\begin{align}
\max_{p_k,\,\nu_k}\quad 
& J_k \;=\;   p_k\,\alpha_{k+1}\,z_{k+1}\ \\[1mm]
\text{s.t.}\quad \nonumber
& \eta_1^*(\mathbf{x}_{k+1},\,p_k,\,\nu_k)\;\ge\;0,\\ \nonumber
& \eta_2^*(\mathbf{x}_{k+1},\,p_k,\,\nu_k)  \geq 0\\ \nonumber
& \eqref{eqn:service_rate}, \eqref{eqn: Unfairness Index}, \eqref{eqn: aux_alpha}, \eqref{eqn: dynamics of new states}, \eqref{eqn: iHOCBF} 
\end{align}

\begin{proposition}[Robust fairness and capacity guarantees]\label{prop:robust_forward_invariance}
Consider the system with state defined in \eqref{eqn: extended state vector} and dynamics given in
\eqref{eqn: dynamics of new states}. Let $\mathcal{W}(t;p)$ denote the set of user-type
distributions consistent with the observable aggregate measurements at time $t$ under price $p$.
Assume that the initial state satisfies $q(0)\le q_{\max}$ and $I(0)\le \theta_d$. For any true
(unobservable) demand composition that satisfies $\mathbf{w}(t)\in\mathcal{W}(t;p)$, $t\ge 0$, if the control inputs $(p,\nu)$ are generated by the proposed controller as a feasible solution of
\textbf{Problem~4} based on the current observations, then the resulting system satisfies
$q(t)\le q_{\max},~I(t)\le \theta_d,~\forall t\ge 0.$
\end{proposition}

\begin{proof}
Let $\mathbf{w}^*(t)$ denote the true (unobservable) user-type composition at time $t$.
For any control inputs $(p,\nu)$ satisfying $\eta_j^*(\mathbf{x},p,\nu)\ge 0$ for $j\in\{1,2\}$, it follows from the worst-case margins defined in  \eqref{eqn: eta_1_star} and \eqref{eqn: eta_2_star} that 
\[
\eta_j(\mathbf{x},p,\nu,\mathbf{w})\ge 0, \quad \forall \mathbf{w}\in\mathcal{W}(t;p),\quad j \in \{1,2\}
\]
Since the true composition satisfies $\mathbf{w}^*(t)\in \mathcal{W}(t;p)$, we obtain
\[
\eta_j(\mathbf{x},p,\nu,\mathbf{w}^*(t))\ge 0, \quad j \in \{1,2\}
\]

For the service queue capacity constraint, $\eta_1$ corresponds to the iHOCBF condition
associated with the barrier function $b(x)=q_{\max}-q$. By the forward-invariance property
of iHOCBFs, satisfaction of the initial condition $q(0)\le q_{\max}$ together with $\eta_1\ge 0$ implies that
$q(t)\le q_{\max}$ for all $t\ge 0$.
For the fairness constraint, $\eta_2$ is defined as $\eta_2(\mathbf{x},p,\nu,\mathbf{w})=\theta_d-I(t)$.
Thus, $\eta_2(\mathbf{x},p,\nu,\mathbf{w}^*(t))\ge 0$ directly implies that $I(t)\le \theta_d$ for all
$t\ge 0$.
\hfill 
\end{proof}

\section{Simulation Results}
This section presents simulation results that demonstrate the effectiveness of the proposed dynamic access pricing and admission control framework under several representative scenarios. 
Unless otherwise stated, we consider a system with two user types. Group~1 corresponds to price-elastic users, characterized by the price sensitivity parameters \(r_{1,1}=0.05\) and \(r_{1,2}=0\) in~\eqref{eqn:dropout_rate}, while Group~2 represents inelastic users with \(r_{2,1}=0\) and \(r_{2,2}=0\). User arrivals are generated according to a clipped Gaussian distribution to capture realistic demand variability while ensuring bounded arrival rates. We note that the proposed framework is not restricted to this specific distribution and can accommodate general stochastic arrival processes.

For comparison purposes, we consider a baseline \emph{Surge Pricing Policy} in which the access price increases monotonically and the admission rate decreases monotonically as the service queue length grows. The specific pricing and admission rate functions for surge pricing are adopted from our prior work \cite{king2025dynamic} for better comparison with the analysis in \cite{king2025dynamic}.
Specifically, the surge pricing baseline adjusts the access price according to the normalized service congestion level
\begin{equation}
\rho \triangleq \min\!\left(1,\; \frac{q}{q_{\max}}\right),
\end{equation}
and computes the price as
\begin{equation}
p^{surge}(q)
=
\max\!\left(
p_{\min},\;
\min\!\left(
p_{\max},\;
1 + 9 \bigl( 3 \rho^2 - 2 \rho^3 \bigr)
\right)
\right).
\end{equation}
where $q_{\max}$ denotes the service queue capacity and $p_{\min}, p_{\max}$ are the minimum and maximum allowable prices, respectively, where $p_{\min}=0$, $p_{\max}=10$ in our simulation.
Under the surge pricing baseline, the admission rate is given by
\begin{equation}
\alpha^{surge}(q)
=
\max\!\left(
0,
\min\!\left(
1,\;
1 + b_1 \rho + b_2 \rho^2
- b_3\rho^3
\right)
\right)
\end{equation}
where \(b_1\) and \(b_2\) are design parameters shaping the sensitivity of admissions to congestion, where $b_1 = -0.129$, $b_2 = -0.967$ and $b_3=-0.096$ in our case.

The decision window is set to $T_d=0.1$. At each decision point, prices and admission rates are computed using the proposed dynamic fair pricing controller and the baseline surge pricing policy. We compare their performance under light, heavy, and time-varying demand, examine the resulting policy and performance trends, and evaluate robustness under the limited observability assumption.

\subsection{Light Demand Scenario}
The simulation results under the light-demand setting with mean arrival rates $\bar K_1 = 4$ for elastic users and $\bar K_2 = 2$ for inelastic users are shown in Fig. \ref{fig:compare_light_Gaussian}. 
All quantities are shown as instantaneous time trajectories. In particular, the ``Price \& Accumulated Revenue'' subplot depicts the instantaneous price $p(t)$ and accumulated revenue rate $r(t)=\sum_t{p(t)\alpha(t)z(t)}$.
It can be observed that both the proposed controller and the baseline surge pricing policy stabilize the system and satisfy the fairness constraint, with the unfairness index remaining below the prescribed threshold. However, their performance differs at steady state. The proposed dynamic fair pricing controller maintains lower queue levels while sustaining a high admission rate, leading to higher revenue, which aligns with its control objective. These results demonstrate that, under light traffic conditions, the proposed method attains improved revenue performance without compromising system stability or fairness.

\begin{figure*}[t]
    \centering
    \includegraphics[width=1.99\columnwidth]{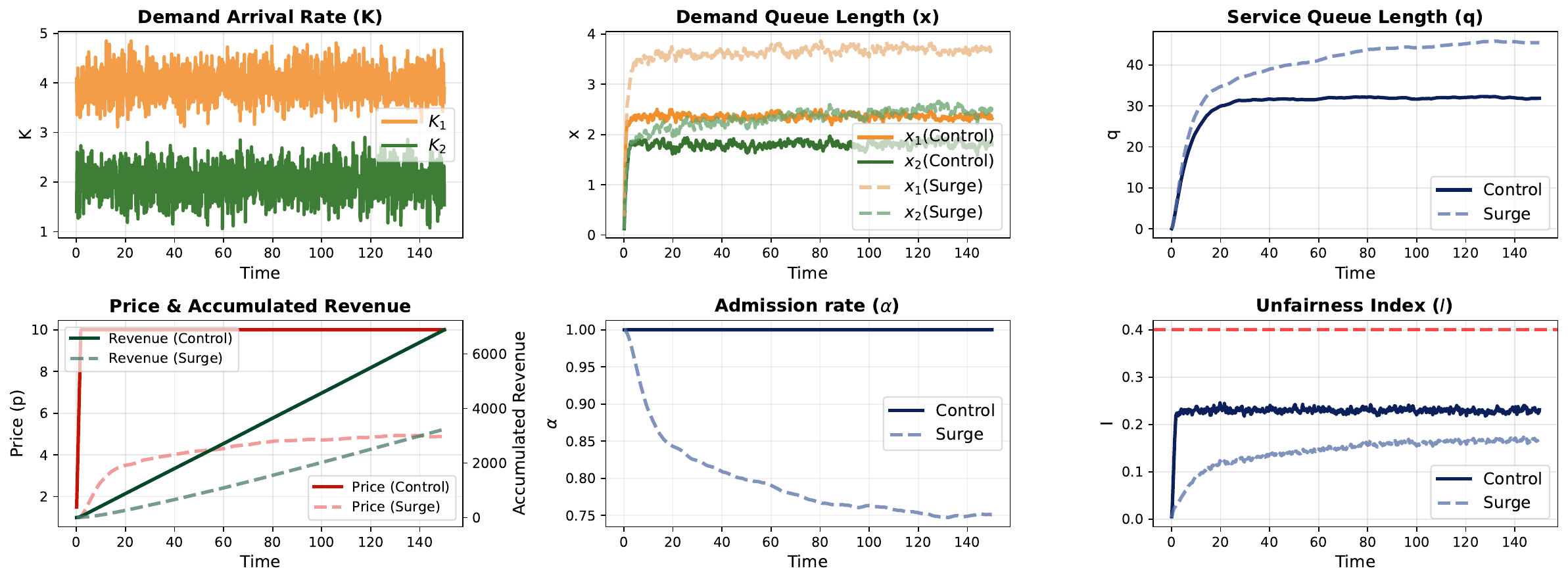}
    \caption{Trajectories Comparison under Light Stochastic Gaussian Demand}   
    \label{fig:compare_light_Gaussian}
\end{figure*}

\subsection{Heavy Demand Scenario}\label{subsec: heavy demand scenario}
The heavy-demand scenario with $\bar K_1 = 7$ and $\bar K_2 = 4$ is illustrated in Fig.~\ref{fig:compare_heavy_Gaussian}. In this regime, the baseline surge pricing policy stabilizes the demand queue primarily by continuously increasing prices to suppress elastic demand. While effective for congestion control, this strategy leads to a rapid increase in the unfairness index, which exceeds the prescribed threshold and thus violates the fairness constraint. In contrast, the proposed optimal dynamic fair pricing controller responds to system overload by lowering prices and avoiding aggressive demand exclusion, thereby maintaining the unfairness index within the allowable bound. Although this results in growing demand queues during the peak period, it preserves equitable access across user groups. Since heavy demand is often transient (e.g., due to temporary disruptions such as severe weather), buffering demand while maintaining fairness is preferable to selectively excluding elastic users through excessive pricing. Even for sustained high-demand conditions, these results further suggest that capacity expansion, rather than discriminatory pricing, is a more appropriate long-term solution.

\begin{figure*}[t]
    \centering
    \includegraphics[width=1.99\columnwidth]{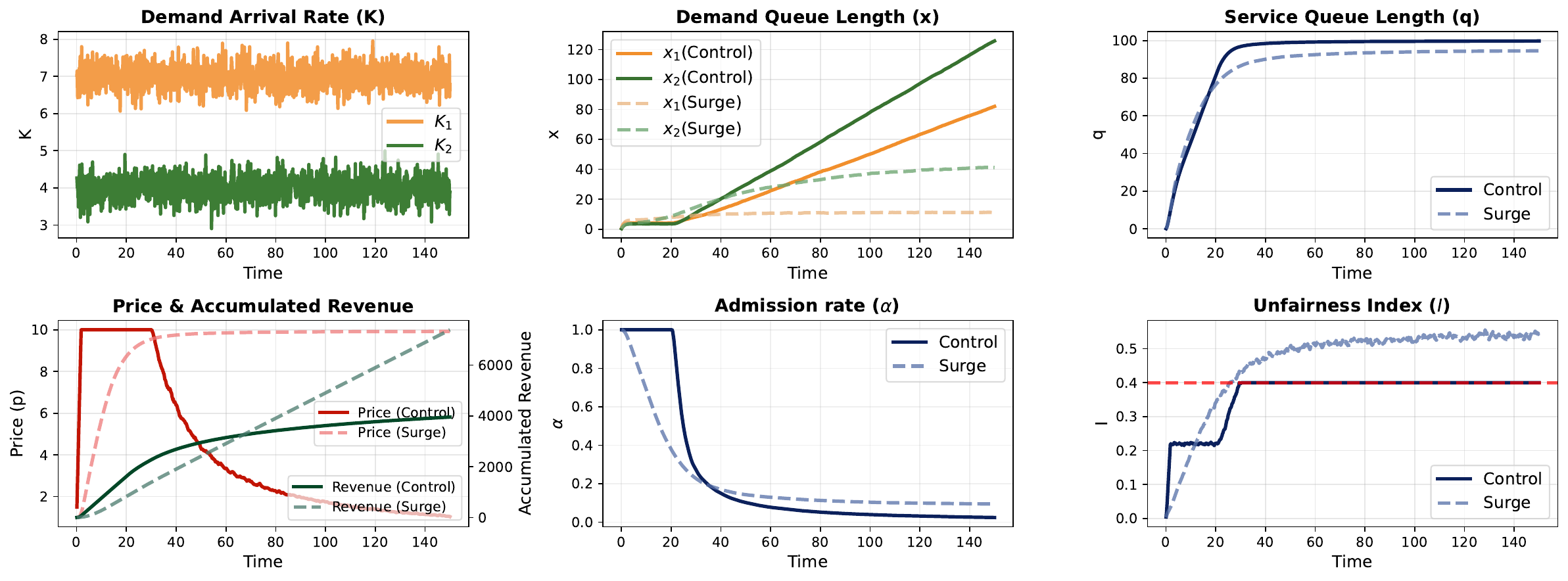}
    \caption{Trajectories Comparison under Heavy Stochastic Gaussian Demand}   
    \label{fig:compare_heavy_Gaussian}
\end{figure*}

\subsection{Dynamic Stochastic Demand}
The performance under time-varying demand is shown in Fig.~\ref{fig:compare_dynamic_Gaussian}. In this scenario, the inelastic demand (group~2) experiences a sharp increase at $t=60$ followed by a gradual decay around $t=100$, while the elastic demand (group~1) remains elevated over the interval $t\in[60,180]$. When faced with this unexpected demand surge, the proposed dynamic fair pricing controller adapts by jointly reducing the price and admission rate, preventing violation of the service queue capacity constraint while maintaining the unfairness index within the prescribed threshold. In contrast, the baseline surge pricing policy primarily relies on increasing prices to suppress elastic demand, which leads to a pronounced rise in the unfairness index during $t\in[80,130]$. 
Notably, although our controller sets lower prices during congestion, it still achieves higher cumulative revenue than the surge pricing policy because the congestion period is short, as is typical of many unexpected demand surges.
These results demonstrate that the proposed method can effectively accommodate time-varying demand while preserving fairness and stable system performance. 


\begin{figure*}
    \centering
    \includegraphics[width=1.99\columnwidth]{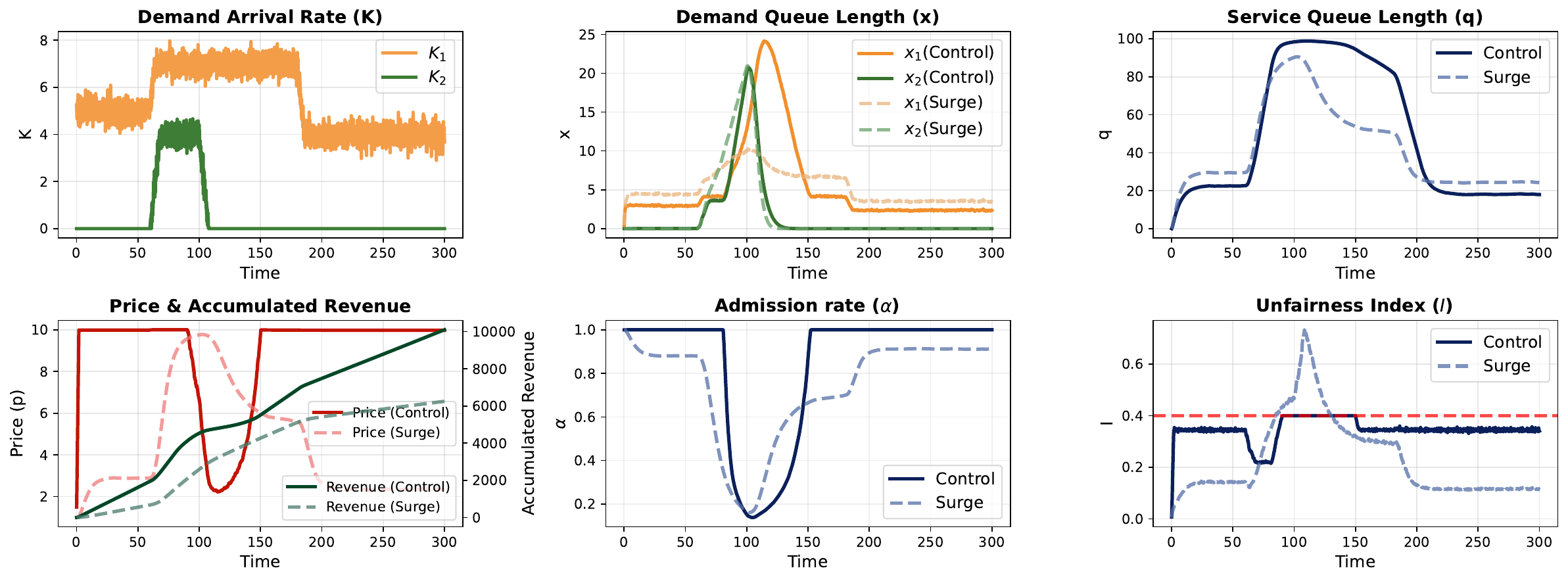}
    \caption{Trajectories Comparison under Time-varying Gaussian Demand}
    \label{fig:compare_dynamic_Gaussian}
\end{figure*}

\subsection{Performance across Different Unfairness thresholds ($\theta_d$)}
Figure \ref{fig:compare theta_d} further illustrates the trade-off induced by the unfairness threshold under the heavy-demand setting described in Subsection\ref{subsec: heavy demand scenario}. When the threshold is tight ($\theta_d=0.2$), the fairness constraint becomes active early, causing the controller to sharply reduce both price and admission rate after the congestion onset. As a result, the unfairness index stays close to the prescribed bound, but cumulative revenue is the lowest among the three cases. As the threshold is relaxed to $\theta_d=0.4$ and $\theta_d=0.6$, the controller can maintain higher prices and larger admission rates for a longer period, which increases cumulative revenue while allowing a correspondingly higher, but still controlled, level of unfairness. In particular, the case $\theta_d=0.6$ produces revenue higher than the baseline surge pricing policy, while still keeping the unfairness index below the prescribed limit throughout the horizon. These results show that $\theta_d$ provides a transparent tuning parameter for balancing fairness and revenue under heavy traffic.

\begin{figure*}
    \centering
    \includegraphics[width=1.99\columnwidth]{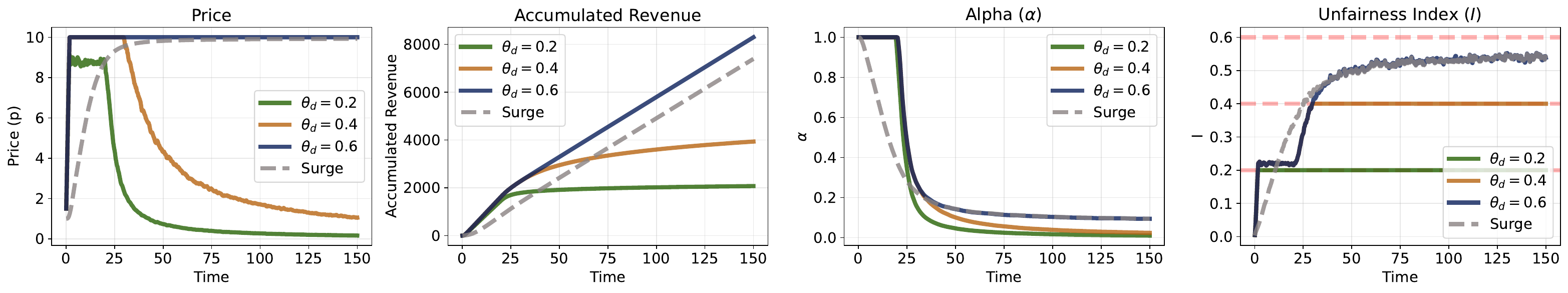}
    \caption{Trajectories Comparison under different $I$ threshold ($\theta_d$)}
    \label{fig:compare theta_d}
\end{figure*}

\subsection{Performance across Different Demand}
To examine the control pattern of the proposed dynamic fair pricing controller, we fix the inelastic demand at $K_2=2$ and vary the elastic demand $K_1$ from $3$ to $10$, running each scenario for $100$ time steps. The resulting control actions and performance metrics are shown in Fig.~\ref{fig:compare_different_demand_K2=2}. When demand is low ($K_1 \leq 6$), the controller sets relatively high prices, stabilizing both demand and service queues well below capacity. As $K_1$ increases further, the service queue approaches its capacity and the fairness constraint becomes active, triggering a sharp reduction in price and admission rate. This behavior contrasts with the surge pricing policy, which monotonically increases prices in response to growing service queues, leading to violations of the unfairness threshold. Notably, the resulting non-monotone pricing pattern, which includes high prices under light load followed by a sharp drop under heavy load, aligns with the structure identified in prior work~\cite{king2025dynamic}, where stability properties of such pricing functions were analytically established.

\begin{figure*}
    \centering
    \includegraphics[width=1.99\columnwidth]{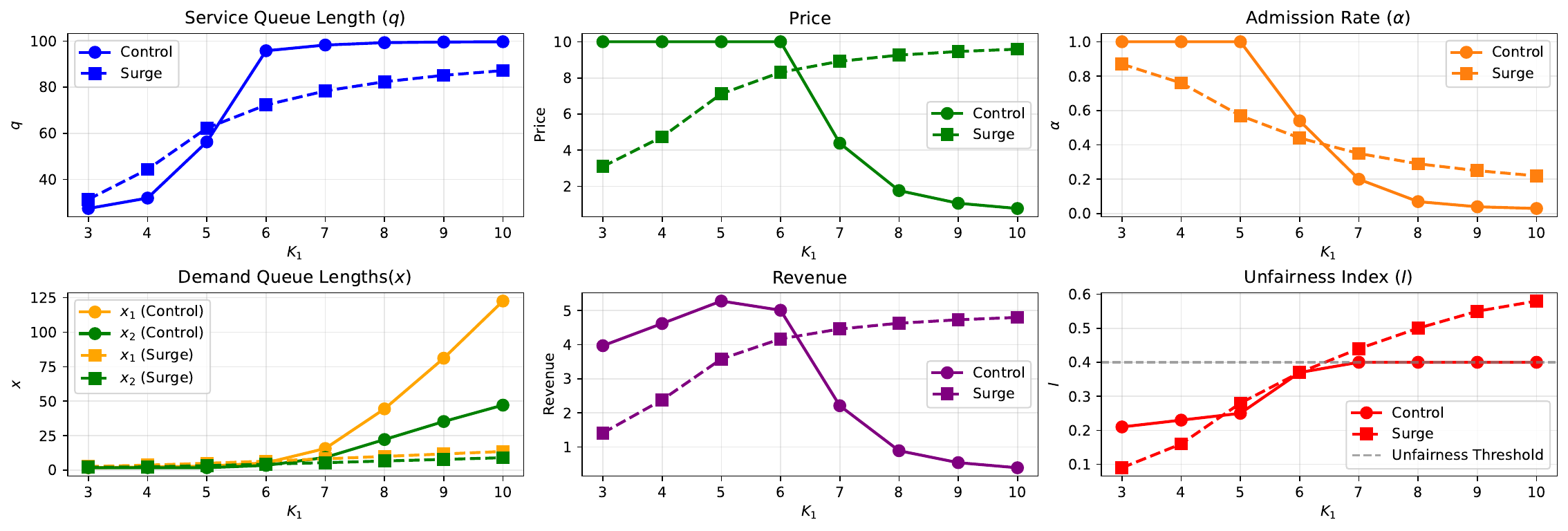}
    \caption{Performance across Different $K_1$ with $K_2=2$}
    \label{fig:compare_different_demand_K2=2}
\end{figure*}

\subsection{Performance with Unobservable States}
The above simulations consider two user groups, so the demand composition is uniquely determined. We next consider a more challenging three-group setting, where the demand composition cannot be uniquely identified from the available measurements and therefore requires the robust control framework introduced in Section~\ref{sec: robust pricing control}. 
 We consider three user groups with elasticity parameters $r_{1,1}=0.05$, $r_{2,1}=0.02$, and $r_{3,1}=0$, and $r_{1,2}=r_{2,2}=r_{3,2}=0$, corresponding to highly elastic, moderately elastic, and inelastic users, respectively. The time-varying demand profile is shown in Fig.~\ref{fig:unobservable_multi}(a), where group~2 and group~3 experience demand surges around $t=60$ and $t=150$, respectively, followed by a decay near $t=200$. Although these group-wise states are unobservable to the controller, the proposed framework estimates worst-case demand queue lengths for each group, as illustrated in Fig.~\ref{fig:unobservable_multi}(b), where solid curves denote the true (unobservable) states and dashed curves represent the estimated averages. Using these worst-case estimates, the controller computes robust pricing and admission decisions that satisfy both the service queue capacity and fairness constraints. As shown in Fig.~\ref{fig:unobservable_multi}(f), when the fairness constraint becomes active under estimated states for $t\in[155,245]$, the unfairness index evaluated using the true states remains below the threshold with a small margin.
In contrast, the baseline surge pricing policy leads to a clear violation of the fairness constraint. This demonstrates that the proposed robust fair pricing framework can effectively guarantee fairness even under unobservable and time-varying demand conditions.

\begin{figure*}
    \centering
    \includegraphics[width=1.99\columnwidth]{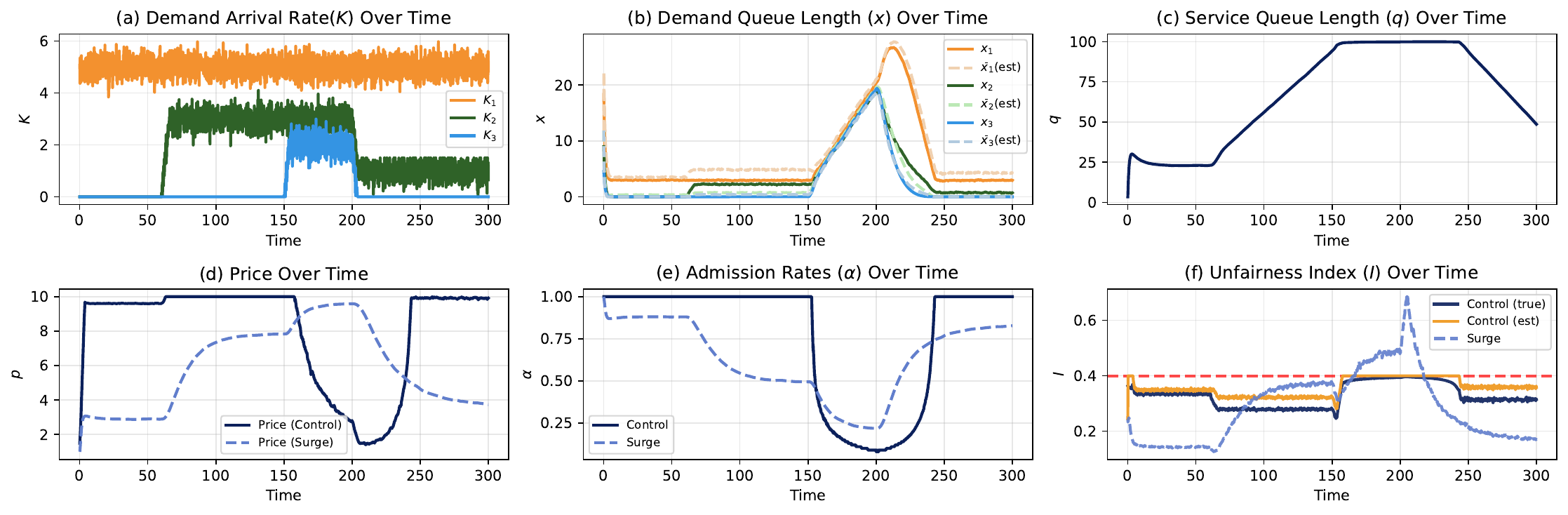}
    \caption{Performance with Unobservable States (3 groups)}
    \label{fig:unobservable_multi}
\end{figure*}

\section{conclusion}
We have presented a robust dynamic pricing and admission control framework for congested service systems with heterogeneous and partially observable user demand. By treating pricing and admission as control inputs and enforcing service capacity through an integral high-order CBF, the proposed approach guarantees safe queue operation while preserving actuation authority of both controls. Robustness to unobservable demand composition was achieved via a worst-case formulation that enforces capacity and fairness constraints over all user-type distributions consistent with aggregate measurements.  Simulation results under stochastic and time-varying demand scenarios demonstrated that the proposed method improves fairness and revenue compared to standard surge pricing policies, while maintaining stable and safe system operation.

\bibliographystyle{IEEEtran}
\bibliography{ref}
\end{document}